\documentclass[a4paper,UKenglish,cleveref,autoref]{lipics-v2019}
\usepackage{amsmath,amssymb,bbm,amsthm}
\usepackage{thm-restate,color,xcolor,xspace}
\usepackage{hyperref,cleveref}
\usepackage{mathtools}
\usepackage{algorithm}
\usepackage[noend]{algpseudocode}

\newcommand{\f}{\frac}
\newcommand{\cd}{\cdot}
\newcommand{\bn}{\binom}
\newcommand{\sr}{\sqrt}
\newcommand{\cds}{\cdots}
\newcommand{\lds}{\ldots}

\newcommand{\s}{\subseteq}

\newcommand{\calC}{\mathcal{C}}
\newcommand{\calP}{\mathcal{P}}

\newcommand{\BE}{\begin{enumerate}}
\newcommand{\EE}{\end{enumerate}}
\newcommand{\im}{\item}
\newcommand{\BI}{\begin{itemize}}
\newcommand{\EI}{\end{itemize}}

\newcommand{\Prod}{\displaystyle\prod\limits}

\newcommand{\logn}{\log n}

\newcommand{\inv}{^{-1}}

\newcommand{\R}{\mathbb R}
\newcommand{\Z}{\mathbb Z}
\newcommand{\F}{\mathbb F}

\newcommand{\eps}{\epsilon}
\newcommand{\e}{\epsilon}
\newcommand{\de}{\delta}
\newcommand{\De}{\Delta}
\newcommand{\la}{\lambda}
\newcommand{\g}{\gamma}

\newcommand{\al}{\alpha}

\newcommand{\Om}{\Omega}
\newcommand{\el}{\ell}

\newcommand{\Th}{\Theta}

\newcommand{\opt}{\text{OPT}}

\newcommand{\lf}{\lfloor}
\newcommand{\rf}{\rfloor}
\newcommand{\lc}{\lceil}
\newcommand{\rc}{\rceil}

\newcommand{\E}{\mathbb E}

\newcommand{\1}{\mathbbm 1}
\newcommand{\poly}{\textup{poly}}

\newcommand{\lp}{\left(}
\newcommand{\rp}{\right)}
\newcommand{\lb}{\left[}
\newcommand{\rb}{\right]}
\newcommand{\lmt}{\left[\begin{matrix}}
\newcommand{\rmt}{\end{matrix}\right]}

\newtheorem{fact}[theorem]{Fact}
\newtheorem{observation}[theorem]{Observation}
\newtheorem{subclaim}[theorem]{Subclaim}
\newcommand{\BT}{\begin{theorem}}
\newcommand{\ET}{\end{theorem}}
\newcommand{\BL}{\begin{lemma}}
\newcommand{\EL}{\end{lemma}}
\newcommand{\BD}{\begin{definition}}
\newcommand{\ED}{\end{definition}}
\newcommand{\BC}{\begin{corollary}}
\newcommand{\EC}{\end{corollary}}
\newcommand{\BO}{\begin{observation}}
\newcommand{\EO}{\end{observation}}
\newcommand{\BF}{\begin{fact}}
\newcommand{\EF}{\end{fact}}
\newcommand{\BCL}{\begin{claim}}
\newcommand{\ECL}{\end{claim}}
\newcommand{\BSCL}{\begin{subclaim}}
\newcommand{\ESCL}{\end{subclaim}}
\newcommand{\BP}{\begin{proof}}
\newcommand{\EP}{\end{proof}}
\newcommand{\BPS}{\begin{proof}[Proof (Sketch)]}
\newcommand{\EPS}{\end{proof}}
\Crefname{observation}{Observation}{Observations}
\Crefname{assumption}{Assumption}{Assumptions}
\Crefname{reduction}{Reduction}{Reductions}
\Crefname{claim}{Claim}{Claims}
\Crefname{subclaim}{Subclaim}{Subclaims}

\renewcommand{\paragraph}{\subsubsection}

\newcommand{\para}{\paragraph}

\newcommand{\cost}{\textup{\textsf{cost}}\xspace}

\newcommand{\flowi}{\textup{\textsf{FlowInstance}}\xspace}

\newcommand{\bd}{\mathbf{d}}

\newcommand{\ball}{\textup{\textsf{ball}}\xspace}
\newcommand{\Bern}{\textup{\textsf{Binomial}}\xspace}

\newcommand{\kmed}{\textup{\textsf{CapKMed}}\xspace}
\newcommand{\kmeans}{\textup{\textsf{CapKMeans}}\xspace}
\newcommand{\fkmed}{\textup{\textsf{FracCapKMed}}\xspace}
\newcommand{\fkmeans}{\textup{\textsf{FracCapKMeans}}\xspace}

\usepackage{graphicx}

\title{On the Fixed-Parameter Tractability of Capacitated Clustering}

\author{Vincent Cohen-Addad}{CNRS \& Sorbonne Université}{}{}{}{}
\author{Jason Li}{Carnegie Mellon University}{}{}{Supported in part by NSF awards CCF-1536002, CCF-1540541, and CCF-1617790.}

\authorrunning{V.\,Cohen-Addad and J.\,Li}
\Copyright{Vincent Cohen-Addad and Jason Li}



\EventEditors{John Q. Open and Joan R. Access}
\EventNoEds{2}
\EventLongTitle{42nd Conference on Very Important Topics (CVIT 2016)}
\EventShortTitle{CVIT 2016}
\EventAcronym{CVIT}
\EventYear{2016}
\EventDate{December 24--27, 2016}
\EventLocation{Little Whinging, United Kingdom}
\EventLogo{}
\SeriesVolume{42}
\ArticleNo{23}

\begin{document}
\nolinenumbers
\maketitle

\begin{abstract}
We study the complexity of the classic capacitated $k$-median and $k$-means problems parameterized by the number of centers, $k$.  These problems are notoriously difficult since the best known approximation bound for high dimensional Euclidean space and general metric space is $\Theta(\log k)$ and it remains a major open problem whether a constant factor exists.
  
We show that there exists a $(3+\epsilon)$-approximation algorithm for the capacitated $k$-median and a $(9+\epsilon)$-approximation algorithm for the capacitated $k$-means problem in general metric spaces whose running times are $f(\epsilon,k) n^{O(1)}$.  For Euclidean inputs of arbitrary dimension, we give a $(1+\epsilon)$-approximation algorithm for both problems with a similar running time.  This is a significant improvement over the $(7+\epsilon)$-approximation of Adamczyk et al. for $k$-median in general metric spaces and the $(69+\epsilon)$-approximation of Xu et al. for Euclidean $k$-means.
\end{abstract}

\section{Introduction}
Clustering under capacity constraints is a fundamental problem whose
complexity is still poorly understood. The capacitated $k$-median and
$k$-means problems have attracted a lot of attention over
the recent years (\textit{e.g.}: ~\cite{byrka2016approximation,DBLP:conf/soda/Li15,DBLP:conf/soda/Li16,DBLP:journals/talg/Li17,DBLP:conf/icalp/DemirciL16,byrka2014bi,Chuzhoy:2005:AKM:1070432.1070569,charikar2002constant}),
but the best known approximation algorithm for
capacitated $k$-median remains a somewhat folklore
$O(\log k)$-approximation using the classic technique of
embeddings the metric space into trees that follows from
the work of Charikar et al~\cite{charikar1998rounding}
on the uncapacitated version,
see also~\cite{2018arXiv180905791A} for a complete exposition.

Arguably, the hardness of the problem comes from having both a hard
constraint on the number of clusters, $k$, and on the number of
clients that can be assigned to each cluster.
Indeed, constant factor approximation algorithms are known if the
capacities~\cite{DBLP:conf/soda/Li15,DBLP:conf/soda/Li16}
or the number of clusters can be violated by a $(1+\eps)$
factor~\cite{byrka2016approximation,DBLP:conf/icalp/DemirciL16},
for constant $\eps$. Moreover, the capacitated facility location
problem admits constant factor approximation algorithms with no capacity
violation.
On the other hand and perhaps surprisingly, the best known
lower bound for capacitated $k$-median is not higher than the
$1+2/e$ lower bound for the uncapacitated version of the problem.

Thus, to improve the 
understanding of the problem a natural direction
consists in obtaining better approximation
algorithms in some specific metric spaces,
or through the fixed-parameter complexity 
of the problem. 
For example, a quasi-polynomial time approximation
scheme (QPTAS) 
for capacitated $k$-median in Euclidean space of fixed dimension
with $(1+\eps)$ capacity
violation was known since the late 90's~\cite{ARR98}. 
This has been recently improved to a PTAS for $\R^2$ and a
QPTAS for doubling metrics without
capacity violation~\cite{abs-1812-07721}.
It remains an interesting
open question to obtain
constant factor approximation for other metrics such as
planar graphs or Euclidean
space of arbitrary dimension.

For many optimization problems are at least  W[1]-hard and so
obtaining exact fixed-parameter tractable (FPT) algorithms
is unlikely. However, 
FPT algorithms have recently
shown that they can help break long-standing barriers
in the world of approximation
algorithms.
FPT approximation algorithms achieving better approximation
guarantees than the best known polynomial-time approximation
algorithms for some classic
W[1]- and W[2]-hard problems have been designed.
For example, for $k$-cut~\cite{Gupta:2018:FAB:3174304.3175483},
for $k$-vertex separator~\cite{Lee2018}
or $k$-treewidth-deletion~\cite{DBLP:conf/soda/GuptaLLM019}.

For the fixed-parameter tractability of the $k$-median and $k$-means
problems, a
natural parameter is
the number of clusters $k$. The FPT complexity of the classic
uncapacitated $k$-median problem, parameterized by $k$, has
received a lot of attention over the last 15 years.
From a lower bound perspective,
the problem is known to be W[2]-hard in general metric spaces and
assuming the exponential time hypothesis (ETH), even for points in $\R^4$,
there is no exact algorithm running in time
$n^{o(k)}$~\cite{Cohen-AddadMRR18}. For
$\R^2$ there exists an exact $n^{O(\sqrt{k})}$ which is the best one can
hope for assuming ETH~\cite{Cohen-AddadMRR18}, see also~\cite{MarxP15}.

From an upper bound perspective, \emph{coreset}
constructions and PTAS with running time $f(k,\epsilon) n^{O(1)}$ have
been known since the early
00's~\cite{VegaKKR03,Kumar2004,HaK07,HaM04,FS05}.
In the language of fixed-parameter tractability,
a coreset is essentially an ``approximate kernel'' for the problem: given
a set $P$ of $n$ points
in a metric space, a coreset is, loosely speaking, 
a mapping from the points in $P$ to a set of points $Q$ of size
$(k \log n \eps^{-1})^{O(1)}$ such that any clustering of $Q$ of cost
$\gamma$ can be converted into a clustering of $P$
of cost at most $\gamma \pm \eps \cost(\opt)$, through the inverse
of the mapping (where $\opt$ is the optimal solution for $P$).
See Definition~\ref{def:coreset}
for a more complete definition. 

In Euclidean space, several coreset
constructions for uncapacitated $k$-median
are independent of the input size and of the dimension and
so are truly approximate kernels. Thus
approximation schemes can simply be obtained by enumerating all possible
partitions of the coreset points into $k$ parts, evaluating the cost
of each of them and outputing the one of minimum cost.
However, obtaining similar results in general metric spaces seems much
harder and is likely impossible. In fact, obtaining
an FPT approximation algorithm with approximation guarantee
less than $1+2/e$ is impossible assuming Gap-ETH, see~\cite{absFPTkmed}.

For the capacitated $k$-median and $k$-means problems much less is known.
First, the coreset constructions or the classic
FPT-approximation schemes techniques of~\cite{KSS10,VegaKKR03}
do not immediately apply. Thus, very little was known until the
recent result of Adamczyk et al.~\cite{2018arXiv180905791A} who proposed a
$(7+\eps)$-approximation algorithm running in time $k^{O(k)} n^{O(1)}$.
More recently, a $(69 +\eps)$-approximation algorithm for
the capacitated $k$-means problem with similar running time has
been proposed by Xu et al.~\cite{abs-1901-04628}.

\subsection{Our Results}
We present a coreset construction for the capacitated $k$-median
and $k$-means problems, with general capacities, and
in general metric spaces (Theorem~\ref{thm:coreset}).
For an $n$ points set, the coreset has size
$\poly(k\e\inv\logn)$.

From this we derive a $(3+\eps)$-approximation for the $k$-median
problem and a $(9+\eps)$-approximation for the $k$-means problem
in general metric spaces.

\begin{theorem}
  \label{thm:generalmetrics}
  For any $\eps >0$, there exists a $(3+\eps)$-approximation algorithm
  for the capacitated $k$-median problem and a $(9+\eps)$-approximation algorithm
  for the capacitated $k$-means problem running in time
  $(k\e\inv\logn)^{O(k)} n^{O(1)}$. This running time can also be bounded by $(k/\e)^{O(k)}n^{O(1)}$.
\end{theorem}

This results in a significant improvement over the recent
results of Adamczyk et al.~\cite{2018arXiv180905791A} for  $k$-median and
Xu et al.~\cite{abs-1901-04628} for (Euclidean) $k$-means, in the same asymptotic running time.

Moreover, combining with the techniques of Kumar et al.~\cite{KSS10},
we obtain a $(1+\eps)$-approximation algorithm for points in $\R^d$,
where $d$ is arbitrary.
We believe that this is an interesting result: 
while it seems unlikely that one can obtain an FPT-approximation
better than $1+2/e$ in general metrics, it is possible to obtain 
an FPT-$(1+\eps)$-approximation in Euclidean metrics of arbitrary
dimension. This works for both the \emph{discrete} and \emph{continuous}
settings: in the former, the set of centers must be chosen from a
discrete set of candidate centers in $\R^d$ and the capacities may not
be uniform, while in the latter the centers can be placed anywhere in
$\R^d$ and the capacities are uniform.

\begin{theorem}
  \label{thm:Euclidean:discrete}
  For any $\eps >0$, there exists a $(1+\eps)$-approximation algorithm
  for the discrete, Euclidean, capacitated $k$-means and $k$-median problems which runs in time 
  $(k\e\inv\logn)^{k\e^{-O(1)}}$ $n^{O(1)}$. This running time can also be bounded by $(k\e\inv)^{k\e^{-O(1)}}$ $n^{O(1)}$.
\end{theorem}

\begin{theorem}
  \label{thm:Euclidean:continous}
  For any $\eps >0$, there exists a $(1+\eps)$-approximation algorithm
  for the continuous, Euclidean, capacitated $k$-means and $k$-median problems running in time
  $(k\e\inv\logn)^{k\e^{-O(1)}}$ $n^{O(1)}$. This running time can also be bounded by $(k\e\inv)^{k\e^{-O(1)}}n^{O(1)}$.
\end{theorem}

These two results are a major improvement over the 69-approximation
algorithm of Xu et al.~\cite{abs-1901-04628}.

\subsection{Preliminaries}
We now provide a more formal definition of the problems.
\BD
Given a set of points $V$ in a metric space with distance function $d$, together with a
set of \emph{clients} $C \subseteq V$, a set of \emph{centers} $\F\s V$ with a \emph{capacity} $\eta_f\in\Z_+$
for each $f \in \F$, and an integer $k$, 
the \emph{capacitated $k$-median problem} asks for a set $F\s \F$ of
$k$ \emph{centers} and an assignment $\mu : C \mapsto F$
such that $\forall f \in F$,  $|\{c \mid \mu(c) = f\}| \le \eta_f$
and that minimizes 
$\sum_{c \in C} d(c, \mu(c))
$. We abbreviate the capacitated $k$-median instance as $((V,d),C,\F,k)$.
\ED
\BD
The \emph{capacitated $k$-means problem} is identical, except we seek to minimize $\sum_{c\in C} d(c,\mu(c))^2$.
\ED

In the literature, centers are sometimes called \emph{facilities}, but we will use \emph{centers} throughout for consistency.

In the case of the capacitated Euclidean $k$-median and $k$-means,
our approach works for the two main definitions.
First, the definition of~\cite{abs-1901-04628,KSS10}: $P = \R^d$ and
capacities are uniform, namely
$\eta_f = \eta_{f'}$, $\forall f,f' \in \R^d$. Second, $P$ is
some specific set of points in $\R^d$, and for each $f \in P$,
the input specifies a specific capacity $\eta_f$

\BD
Given a capacitated $k$-median instance $((V,d),C,\F,k)$ and a set of chosen centers $F\s \F$, define $\kmed(C,F)$ as the cost of the optimal assignment of the clients to the chosen centers. If it is impossible, i.e., the sum of the capacities of the centers is less than $|C|$, then $\kmed(C,F)=\infty$.
\ED

In our analysis, we will also encounter formulations where the clients have positive \emph{real} weights. In this case, we define a \emph{fractional} variant of capacitated $k$-median, where the assignment $\mu$ is allowed to be fractional.

\BD
Suppose the clients also have weights, so we are given clients $C$ and a weight function $w:C\to\R_+$. Let $W\s C\times\R_+$ be the set of pairs $\{(c,w(c)):c\in C\}$. Then, $\fkmed(W,F)$ is the minimum value of $\sum_{c\in C,f\in F}\mu(c,f)\,d(c,f)$ over all ``fractional assignments'' $\mu:C\times F\to\R_+$ such that:
 \BE
 \im $\forall c\in C$, $\sum_{f\in F}\mu(c,f)=w(c)$, i.e., $\mu$ is a proper assignment of clients, and
 \im $\forall f\in F$, $\sum_{c\in C}\mu(c,f)\le\eta_f$, i.e., $\mu$ satisfies capacity constraints at all centers.
 \EE
\ED

\BD
We define $\kmeans(C,F)$ and $\fkmeans(W,F)$ similarly, except our objective functions are $\sum_{c\in C}d(c,\mu(c))^2$ and $\sum_{c\in C,f\in F}\mu(c,f)\,d(c,f)^2$, respectively.
\ED

It is well-known that, given a set $F\s \F$ of centers, the problem of finding the optimum $\mu$ is an (integral) \emph{minimum-cost flow} problem, which can be solved in polynomial time. Therefore, we assume that every time we have a set $F\s\F$, we can evaluate $\kmed(C,F)$ and $\kmeans(C,F)$ in polynomial time. Similarly, $\fkmed$ and $\fkmeans$ can be solved through fractional min-cost flow, or even an LP, in polynomial time. Furthermore, if $W$ is exactly the set $C$ of clients with weight $1$, i.e., $W=\{(c,1):c\in C\}$, then $\kmed(C,F)=\fkmed(W,F)$, since the min-cost flow formulation of $\fkmed$ has integral capacities and therefore integral flows as well.


We now formally state our definition of coresets, sometimes called \emph{strong} coresets in the literature.
\BD
\label{def:coreset}
A (strong) coreset for a capacitated $k$-median instance $((V,d),C,\F,k)$ is a set of weighted clients $W\s C\times \R_+$ such that for every set of centers $F\s\F$ of size $k$,
\[ \fkmed(W,F) \in (1-\e,1+\e) \cd \kmed(C,F) .\]
The definition is identical for capacitated $k$-means, except $\kmed$ and $\fkmed$ are replaced by $\kmeans$ and $\fkmeans$ above.
\ED
\BF
Let $W$ be a coreset for a capacitated $k$-median instance $((V,d),C,\F,k)$. We have
\[ \min_{\substack{F\s\F\\|F|=k}}\fkmed(W,F) \in (1-\e,1+\e) \cd \min_{\substack{F\s\F\\|F|=k}}\kmed(C,F) ,\]
In particular, an $\al$-approximation of $\min_{F\s\F, |F|=k}\fkmed(W,F)$ implies a $(1+O(\e))\al$-approximation to the capacitated $k$-median instance. The same holds in the capacitated $k$-means case, with $\fkmed$ and $\kmed$ replaced by $\fkmeans$ and $\kmeans$, respectively.
\EF

For a capacitated $k$-median or $k$-means instance $((V,d),C,\F,k)$, the \emph{aspect ratio} is the ratio of the maximum and minimum distances between any two points in $C\cup F$. It is well-known that we may assume, with a multiplicative error of $(1+o(1))$ in the optimal solution, that the instance has $\poly(n)$ aspect ratio.\footnote{For example, the following modification to the distances $d$ does the trick. First, compute an $O(\log k)$-approximation~\cite{charikar1998rounding} to the problem, and let that value be $M$. For any two points $u,v\in C\cup F$ with $d(u,v)>Mn^{10}$, truncate their distance to exactly $Mn^{10}$. Then, add $Mn^{-10}$ distance to each pair of points $u,v\in C\cup F$. The aspect ratio is now bounded by $O(n^{20})$.} Therefore, we will make this assumption throughout the paper.

Lastly, we define $\R_+$ and $\Z_+$ as the set of positive reals and positive integers, respectively. As usual, we define \emph{with high probability (w.h.p.)} as with probability $1-n^{-Z}$ for an arbitrarily large positive constant $Z$, fixed beforehand.


\section{Coreset for $k$-median}

In this section, we prove our main technical
result for the $k$-median case: constructing a coreset for capacitated $k$-median of size $\poly(k\logn\,\e\inv)$.


\BT
\label{thm:coreset}
For any small enough constant $\e\ge0$, there exists a Monte Carlo algorithm that, given an instance $((V,d),C,\F,k)$ of capacitated $k$-median, outputs a (strong) coreset $W \s C$ with size $O(k^2\log^2 n/\e^3)$ in polynomial time, w.h.p.
\ET

\BT\label{thm:coreset-means}
For any small enough constant $\e\ge0$, there exists a Monte Carlo algorithm that, given an instance $((V,d),C,\F,k)$ of capacitated $k$-means, outputs a (strong) coreset $W \s C$ with size $O(k^5\log^5n/\e^3)$ in polynomial time, w.h.p.
\ET


Our inspiration for the coreset construction is Chen's algorithm \cite{Che09} based on random sampling. Our algorithm is essentially the same, with slightly worse bounds in the sampling step, although our analysis is a lot more involved. We describe the full algorithm in pseudocode below (see Algorithm~\ref{alg:1}).

At a high level, the algorithm first partitions the client set $C$ into $\poly(k,\logn)$ many subsets, called \emph{rings}, with the help of a polynomial-time approximate solution (see line~\ref{line:1}). The sets are called rings because they are of the form $C_i \cap \lp \ball(f'_i,R) \setminus \ball(f'_i,R/2) \rp$ for some subset of clients $C_i\s C$, some facility $f'_i\in \F$, and some positive number $R$ (see line~\ref{line:rings}). Then, for each ring $C_{i,R}$, if $|C_{i,R}|$ is small enough, the algorithm adds the entire ring into the coreset (each with weight $1$); otherwise, the algorithm takes a random sample of $r=\poly(k,\logn)$ many clients in $C_{i,R}$, weights each sampled client by $|C_{i,R}|/r$, and adds the weighted sample to the coreset. The weighting ensures that the total weight of the sampled points is always equal to $|C_{i,R}|$. To prove that the algorithm produces a coreset w.h.p., Chen union bounds over all $\bn{|\F|}k$ choices of a set of $k$ facilities, and shows that for each choice $F\s\F$, with probability at least $1-n^{-\Om(k)}$, the total cost to assign the coreset points to $F$ is approximately the total cost to assign the original clients $C$ to $F$; this statement is proved through standard concentration bounds. More details and intuition for the algorithm can be found in Section~3 of Chen's paper~\cite{Che09}.

\begin{algorithm}
\caption{CoreSet$(I)$}
\begin{algorithmic}[1]
\State $F'=\{f'_1,\lds,f'_{O(k)}\} \gets$ an $(O(1),O(1))$ bicriteria solution to instance $I$, namely a capacitated $O(k)$-median solution with total cost $ALG'\le O(OPT)$ \Comment{using, e.g., \cite{DBLP:conf/soda/Li16}} \label{line:1}
\State $W \gets \emptyset$ \Comment{$W \s C\times \R_+$ is the final coreset at the end of the algorithm}
\State Define $d_{\min}$ and $d_{\max}$ as the minimum and maximum distances, respectively, between any two points in $C\cup\F$ \Comment{$d_{\max}/d_{\min}$ is the aspect ratio}
\For {each center $f'_i$}\label{line:4} \Comment{$O(k)$ centers}
  \State $C_i \gets$ the clients in $C$ assigned to center $f'_i$
  \For {each $R$, a power of $2$ in the range $[d_{\min},\,2d_{\max}]$} \Comment{$O(\logn)$ iterations, assuming $\poly(n)$ aspect ratio}
      \State $C_{i,R} \gets C_i \cap \lp \ball(f'_i,R) \setminus \ball(f'_i,R/2) \rp$ \Comment{We call the sets $C_{i,R}$ \emph{rings}, with \emph{ring center} $f'_i$. The rings $C_{i,R}$ over all $i,R$ partition the client set $C$.} \label{line:rings}
    \State $r \gets \g k\log n / \e^3$ for sufficiently large (absolute) constant $\g$
    \If {$|C_{i,R}| \le r$} \label{line:12}
      \State add $(c,1)$ to $W$ for each $c\in C_{i,R}$ \Comment{$C_{i,R}$ small enough: add everything into coreset}
    \Else
      \State sample $r$ random centers in $C_{i,R}$ (without replacement) \label{line:15}
      \State add $(c,\f{|C_{i,R}|}r)$ to $W$ for each sampled center $c$ \Comment{weighted so that total weight is still $|C_{i,R}|$} \label{line:16}
    \EndIf
  \EndFor
\EndFor
\end{algorithmic}\label{alg:1}
\end{algorithm}



\subsection{Single ring case} \label{sec:single-ring}

We first restrict ourselves to sampling from a \emph{single} ring $C_{i,R}\s C$. That is, while we still consider the cost of serving the clients outside of $C_{i,R}$, we only perform the sampling (lines~\ref{line:15}--\ref{line:16}) on one ring $C_{i,R}$. The general case of $O(k\log n)$ many rings is more complicated than simply treating each ring separately. Due to space constraints, we only consider the single ring case in this extended abstract, and the rest is deferred to the full version. 

Fix an arbitrary ring $C_{i,R}$ throughout this section, and define $C':=C_{i,R}$ for convenience. Let $N:=|C'|$ be the number of clients, and let $f':=f'_i$ be the ring center of $C'$ (line~\ref{line:4}).
Let $W'$ be the (weighted) centers in $C_{i,R}$ sampled by the algorithm (lines~\ref{line:15}--\ref{line:16}), together with the (unweighted) centers in $C\setminus C'$, which have weight $1$. Our goal is to show that $\fkmed(W',F)$, the cost after sampling only from $C'$, is close to the original cost $\kmed(C,F)$.

\BL\label{lem:single-ring}
W.h.p., for any set of $k$ centers $F\s\F$ satisfying $\kmed(C,F)<\infty$,
\begin{gather}
 | \fkmed(W',F) -  \kmed(C,F) | \le \e NR . \label{eq:med1}
\end{gather}
\EL

It is clear that the output $W$ has size $O(k^2\log^2n/\e^3)$. The rest of this section focuses on proving that $W$ is indeed a coreset, w.h.p.

The intuition behind the $\e NR$ additive error is that we can ``charge'' this error to the cost of the bicriteria solution (line~\ref{line:1}) that $C'$ is responsible for. In particular, the total cost of assigning clients in $C'$ to ring center $f'$ in the bicriteria solution is at least $N\cd R/2$, since all clients in $C'$ are distance at least $R/2$ to $f'$. Therefore, we charge an additive error of $\e NR$ to a $NR/2$ portion of $ALG'$, which is a ``rate'' of $2\e$ to $1$. If we can do the same for all rings, then since the portions of $ALG'$ sum to $ALG'$, our total additive error is at most $2\e \cd ALG' = O(\e)\cd OPT$. Finally, replacing $\e$ with a small enough $\Th(\e)$ gives the desired additive error of $\e\cd OPT$; note that this is where we use that the approximation ratio of $ALG'$ is $O(1)$, and that the specific approximation ratio is not important (as long as it is constant). The formalization of this intuition is deferred to the full version; the argument is identical to Chen's~\cite{Che09}, so we claim no novelty here.


We now prove \Cref{lem:single-ring}.
First of all, if $N=|C'|\le r$ (line~\ref{line:12}), then sampling changes nothing, and $\fkmed(W',F)=\kmed(C,F)$. Therefore, for the rest of the proof, we assume that $N >r= \g k\log n/\e^3$, with the $\g$ taken to be a large enough constant.

Our high-level strategy is the same as Chen's: we union bound over all sets of centers $F\s \F$ of size $k$, and prove that for a fixed set $F$, the probability of violating (\ref{eq:med1}) is at most $n^{-(k+10)}$.\footnote{For simplicity of presentation, we will focus on a success probability of $1-n^{-10}$. The constants can be easily tweaked so that the algorithm succeeds w.h.p., i.e., with probaility $1-n^{-Z}$ for any positive constant $Z$.} Union bounding over all $\le \bn nk$ choices of $F$ gives probability $\le n^{-10}$ of violating (\ref{eq:med1}), proving the lemma. Therefore, from now on, we focus on a single, arbitrary set $F\s\F$ of size $k$ satisfying $\kmed(C,F)<\infty$, and aim to show that (\ref{eq:med1}) fails with probability $\le n^{-(k+10)}$.

For our analysis, we define a function $g:\R^{C'}_+\to\R_+$ as follows. For an input vector $\bd\in\R_+^{C'}$ (indexed by clients in $C'$), consider a min-cost flow instance $\flowi(\bd)$ on the graph metric with the following demands: set demand $d_c$ at each client $c\in C'$, demand $1$ at each client $c\in C\setminus C'$, and demand $N-\sum_{c\in C'}d_c$ (this demand can be negative) at ring center $f'=f'_i$ (so we are effectively treating $f'$ as a special client with possibly negative demand, not a facility). Observe that  $\flowi(\bd)$  is a feasible min-cost flow instance, because the sum of demands is exactly
\[ \sum_{c\in C'} d_c + |C\setminus C'| + \lp N-\sum_{c\in C'}d_c\rp = |C\setminus C'| + N = |C|,\]
which is the same as the sum of demands in the instance $\kmed(C,F)$, which is feasible by assumption.

Given this setup for an input vector $\bd\in\R^{C'}_+$, we define the function $g(\bd)$ as the min-cost flow of $\flowi(\bd)$. Observe that $g(\1)$ is exactly $\kmed(C,F)$.

Now define a random vector $X\in \R_+^{C'}$ as follows. Each coordinate of $X$ is independently $N/r$ with probability $r/N$ and $0$ otherwise, so that $\E[X]=\1$. Note that $X$ does not accurately represent our sampling of $r$ clients, since this process is not guaranteed to sample exactly $r$ clients. Nevertheless, it is intuitively clear that with probability $\Om(1/n)$, $X$ will indeed have exactly $r$ nonzero entries, since $r$ is the expected number; we prove this formally in the following simple claim (with $p=r/N$), 
whose routine proof is deferred to the full version. And if we \emph{condition} on this event, then $g(X)$ and $\kmed(C,F)$ are now identically distributed. 

\begin{restatable}{claim}{CoinFlip} \label{clm:CoinFlip}
        Let $N$ be a positive integer, and let $p\in(0,1)$ such that $pN$ is an integer. The probability that $\Bern(N,p)=pN$ is at least $\Om(1/\sr{N})$.
\end{restatable}

In light of all this, our main argument has two steps. First, we show that $g(X)$ is concentrated around $\E[g(X)]$ using martingales. However, what we really need is concentration around $g(\E[X])=g(\1)=\kmed(C,F)$, so our second step is to show that $\E[g(X)] \approx g(\E[X])$ (with probability $1$). We formally state the lemmas below which, as discussed, together imply \Cref{lem:single-ring}.
\BL\label{lem:step1}
Assume that $|C'| > \Th(k\log n/\e^3)$. With probability $\ge 1-n^{-(k+20)}$, we have $| g(X) - \E[g(X)] | \le \e NR/2$.
\EL
\BL\label{lem:step2}
Assume that $|C'| > \Th(k\log n/\e^3)$. Then, $| \E[g(X)] - g(\E[X]) | \le \e NR/2$.
\EL




\para{Proof of \Cref{lem:step1}: concentration around $\E[g(X)]$ via martingales.}
To show that $g(X)$ is concentrated around its mean, we show that $g$ is sufficiently Lipschitz (w.r.t.\ the $\el_1$ distance in $\R_+^{C'}$), and then apply standard martingale tools.

\BCL\label{clm:lip}
The function $g$ is $R$-Lipschitz w.r.t.\ the $\el_1$ distance in $\R_+^{C'}$.
\ECL
\BP
Fix a client $c\in C'$, and consider two vectors $\bd,\bd'\in\R_+^{C'}$ with $\bd'=\bd+\de\cd\1_c$. By definition of $\flowi$, the only difference between $\flowi(\bd)$ and $\flowi(\bd')$ is that in $\flowi(\bd')$, client $c$ has $\de$ more demand and ``special client'' $f'$ has $\de$ less demand. Therefore, if we begin with the min-cost flow of $\flowi(\bd)$, and then add $\de$ units of flow from $c$ to $f'$,  then we now have a feasible flow for $\flowi(\bd')$.\footnote{We define demand so that if a vertex $v$ has $d>0$ demand, then $d$ flow must exit $v$ in a feasible flow, and if it has $d<0$ demand, then $|d|$ flow must enter $v$.} This means that
\[ g(\bd')\le g(\bd) + \de R  .\]
Similarly, starting from a min-cost flow of $\flowi(\bd')$ and then adding $\de$ units of flow from $f'$ to $c$, we obtain a feasible flow for $\flowi(\bd)$, so
\[ g(\bd)\le g(\bd') + \de R.\]
Together, these two inequalities show that $g$ is $R$-Lipschitz.
\EP

We state the following Chernoff bound for Lipschitz functions, which can be proven by adapting the standard (multiplicative) Chernoff bound proof to a martingale. 

\begin{restatable}{theorem}{Martingale}\label{thm:m2}
Let $x_1,\lds,x_n$ be independent random variables taking value $b$ with probability $p$ and value $0$ with probability $1-p$, and let $g:[0,1]^n\to \R$ be a $L$-Lipschitz function in $\el_1$ norm. Define $X:=(x_1,\lds,x_n)$ and $\mu:=\E[g(X)]$. Then, for $0\le\e\le1$:
\[ \Pr\big[\big|g(X)-\E[g(X)] \big| \ge \e pnbL\big] \le 2e^{-\e^2pn/3} \]
\end{restatable}

We apply \Cref{thm:m2} on the $L$-Lipschitz function $g$ with the randomly sampled demands. Set $p:=r/N$ as the sampling probability, so that $X\in\{0,1/p\}^N$ is the random demand vector. Setting $n:=N$, $b:=1/p$, and $L:=R$, we obtain
\begin{align*}
        \Pr&\big[\big|g(X)-\E[g(X)]\big|\ge (\e/2)NR\big] 
        \\&=\Pr \big[ \big| g(X) - \E[g(X)] \big| \ge (\e/2)pnbL ] 
\\&\le 2\exp\lp\f{-(\e/2)^2pn}{3}\rp \\&= 2\exp\lp \f{-(\e/2)^2 (r/N)N}{3} \rp 
= \exp\lp -\Th(\e^2 r) \rp
 = \exp\lp -\Om(\e^2 \cd \f{k\logn}{\e^2}) \rp 
\\&\le n^{-(k+20)} 
\end{align*}
for sufficiently large $\g$ in the definition of $r = \g k\log n / \e^2$. This concludes \Cref{lem:step1}.

\para{Proof of \Cref{lem:step2}: relating $\E[g(X)]$ with $g(\E[X])$.}

We have obtained concentration about $\E[g(X)]$, but we really need concentration around $g(\E[X])=\kmed(C',F)$. We establish this by proving \Cref{lem:step2}.

We first show the easy direction, that $g(\E[X])\le\E[g(X)]$, which essentially follows from the convexity of min-cost flow: Suppose the outcomes of random variable $X$ are $\bd_1,\bd_2,\lds$ with respective probabilities $\mu_1,\mu_2,\lds$, so that $\E[g(X)] = \sum_i \mu_i g(\bd_i)$. Now consider the flow obtained by adding up, for each $i$, the min-cost flow of $\flowi(\bd_i)$ scaled by $\mu_i$. This flow is a feasible flow to $\flowi(\E[X])$ and has cost at most $\E[g(X)]$. Since the min-cost flow of $\flowi(\E[X])$ can only be lower, we have $g(\E[X])\le\E[g(X)]$.

We now prove the other direction: $\E[g(X)] \le g(\E[X]) + \e NR/2$. 

\BCL\label{lem:step2-1}
With probability $1$, $g(X) \le g(\E[X]) + nNR$.
\ECL
\BP
Since $X \in [0,N/r]^N$, and since $g$ is $R$-Lipschitz, the entire range of $g(X)$ is contained in some interval of length $N \cd N/r \cd R \le N \cd n \cd R$. Since $\E[X] \in [0,N/r]^N$ as well, the value $g(\E[X])$ is also contained in that interval. The statement follows.
\EP

\BL\label{lem:step2-2}
With probability $\ge 1-n^{-10}$, $g(X) \le g(\E[X])+0.49\e NR$.
\EL
Due to space constraints, the proof of \Cref{lem:step2-2}, which is long and technical, is deferred to the full version. 
Assuming \Cref{lem:step2-2}, we now show how \Cref{lem:step2-1} and \Cref{lem:step2-2} together imply \Cref{lem:step2}: we have
\begin{align*}
\E[g(X)] &\le n^{-10} \cd \big( g(\E[X])+nNR\big) + (1-n^{-10}) \big( g(\E[X])+0.49\e NR\big) 
\\ &= g(\E[X]) + \big( n^{-10} \cd n + (1-n^{-10}) \cd 0.49\e\big)NR
\\ &\le g(\E[X]) + (\e/2) NR,
\end{align*}
finishing the proof of \Cref{lem:step2}.



\subsection{$(3+\e)$- and $(9+\e)$-approximation -- Proof of Theorem~\ref{thm:generalmetrics}}
In this section, we finish the algorithm for \Cref{thm:generalmetrics}. We will focus mainly on the $k$-median case, since the $k$-means case is nearly identical.

Suppose we run the coreset for the capacitated $k$-median instance with parameter $\e_0$ (to be set later), obtaining a coreset $W\s C\times\R^+$ of size $\poly(k\logn\, \e_0\inv)$. We now want to compute some $F\s\F$ of size $k$ and an assignment $\mu$ of the clients in $W$ to $F$ minimizing $\sum_{(c,w)\in W}w\cd d(c,\mu(c))$. By definition of coreset, if we compute an $\al$-approximation to this problem, then we compute a $(1+\e_0)\al$-approximation to the original capacitated $k$-median problem.

 The strategy is similar to that in~\cite{absFPTkmed}: we guess a set of \emph{leaders} and \emph{distances} that match the optimal solution.
More formally, let $F^*=\{f_1^*,\lds,f_k^*\}\s\F$ be the optimal solution with assignment $\mu^*$. For each $f_i^*\in F^*$, let $(\mu^*)\inv(f_i^*)$ be the clients in the coreset assigned by $\mu^*$ to $f_i^*$, and let $\el_i$ be the client in $(\mu^*)\inv(f_i^*)$ closest to $f_i^*$. We call $\el_i$ the \emph{leader} of the client set $(\mu^*)\inv(f_i^*)$. Also, let $R_i$ be the distance $d(f_i^*,\el_i)$, rounded down to the closest integer power of $(1+\e_1)$ for some $\e_1$ we set later.

The algorithm begins with an enumeration phase. There are $|W|^k$ choices for the set $\{\el_1,\lds,\el_k\}$, and $O(\e_1\inv\logn)^k$ choices for the values $R_1,\lds,R_k$, since we assumed that the instance has aspect ratio $\poly(n)$. So by enumerating over $|W|^kO(\e_1\inv\logn)^k=(k\logn\,\e_0\inv\e_1\inv)^{O(k)}$ choices, we can assume that we have guessed the right values $\el_i$ and $R_i$.

For each leader $\el_i$, define $\F_i$ as the centers $f\in\F$ satisfying $d(\el_i,f) \in [1,1+\e_1)\cd R_i$. Note that $f_i^*\in\F_i$ for each $i$. Next, the algorithm wants to pick the center in each $\F_i$ with the largest capacity. This way, even if it doesn't pick $f^*_i$ for $\F_i$, it picks a center not much farther away that has at least as much capacity.

The most natural solution is to \emph{greedily} choose the center with largest capacity in each $\F_i$. One immediate issue with this approach is that we might choose the same center twice, since the sets $\F_i$ are not necessarily disjoint. Note that this issue is not as pronounced in the uncapacitated $k$-median problem, since in that case, we can always imagine choosing the same center twice and then throwing out one copy, which changes nothing. In the capacitated case, choosing the same center twice effectively doubles the capacity at that center, so throwing out a copy affects the capacity at that center.

One simple fix to this issue is the simple idea of \emph{color-coding}, common in the FPT literature: for each center $f\in \F$, independently assign a uniformly random label in $\{1,2,3,\lds,k\}$. With probability $1/k^k$, each $f^*_i\in F^*$ is assigned label $i$. Moreover, repeating this routine $O(k^k\logn)$ times ensures that w.h.p., this will happen in some iteration. So with a $O(k^k\logn)$ multiplicative overhead in the running time, we may assume that each $f^*_i$ is assigned label $i$.

The algorithm now chooses, from each $\F_i$, the center with the largest capacity among all centers with label $i$. Since $f^*_i$ is an option for each $\F_i$, the center chosen can only have larger capacity. Let the center chosen from $\F_i$ be $f_i$. Let $F:=\{f_1,\lds,f_k\}$ be our chosen centers.

We now claim that $F$ is a $(3+\e_1)$-approximation. Recall $\mu^*$, the optimal assignment to the centers $F^*$; we construct an assignment $\mu$ to $F$ as follows: for each client $c$ in the coreset, if $\mu^*$ assigns $c$ to center $f^*_i$, then we set $\mu(c)=f_i$. Observe that if $\mu^*(c)=f^*_i$, then
\[ d(c,f_i) \le d(c,f^*_i) + d(f^*_i,\el_i) + d(\el_i,f_i) \le d(c,f^*_i)+2(1+\e_1)R_i \le d(c,f^*_i)+2(1+\e_1)\cd d(c,f^*_i) ,\]
where the first inequality follows from triangle inequality, the second follows since both $f^*_i$ and $f_i$ are approximately $R_i$ away from $\el_i$, and the third follows from $d(c,f^*_i)\ge d(\el_i,f^*_i)\ge R$ by our choice of $\el_i$. Therefore, we have $d(c,\mu(c)) = d(c,f_i) \le (3+2\e_1)d(c,f^*_i) = (3+2\e_1)d(c,\mu^*(c))$.
Altogether, the total cost of the assignment $\mu$ is
\[ \sum_{(c,w)\in W}w\cd d(c,\mu(c)) \le \sum_{(c,w)\in W}w\cd (3+2\e_1)d(c,\mu^*(c)) = (3+2\e_1)\,OPT .\]
The optimal assignment can only be better, hence the $(3+2\e_1)$-approximation. This implies a $(1+\e_0)(3+2\e_1)$-approximation in time $\poly(k\log n\,\e_0\inv\e_1\inv)^{O(k)}$. Finally, setting $\e_0,\e_1:=\Th(\e)$, for $\Th(\cd)$ small enough, guarantees a $(3+\e)$-approximation in time $(k\logn\,\e\inv)^{O(k)}n^{O(1)}$.

Lastly, we show that the $(\logn)^{O(k)}$ factor in the running time can be upper bounded by $k^{O(k)}n^{O(1)}$, proving the second running time in \Cref{thm:generalmetrics}. If $k<\f{\logn}{\log\logn}$, then $(\logn)^{O(k)}=(\logn)^{\f{\logn}{\log\logn}}=n^{O(1)}$; otherwise, $k>\f{\logn}{\log\logn} \ge \sr{\logn}$, so $(\logn)^{O(k)} \le (k^2)^{O(k)}$. Therefore, the running time in \Cref{thm:generalmetrics} is at most $O(k/\e)^{O(k)}n^{O(1)}$.

For $k$-means, the algorithm and analysis are identical, except that the total cost is now
\[ _{(c,w)\in W}w\cd d(c,\mu(c))^2 \le \sum_{(c,w)\in W}w\cd \big((3+2\e_1)d(c,\mu^*(c))\big)^2 = (9+O(\e_1))\,OPT , \]
implying a $(9+\e)$-approximation.  This concludes the proof of \Cref{thm:generalmetrics}.

\section{A $(1+\eps)$-Approximation for Euclidean Inputs}

\subsection{The Continuous (Uniform-Capacity) Case --
  Proof of Theorem~\ref{thm:Euclidean:continous}}
In this section we consider the continuous case: namely the
case where centers can be located at arbitary position in $\R^d$
and the capacities are uniform and $\eta \ge n/k$.

Let $\eps>0$. 
Given a set of points $P$, denote by $\opt_1(P)$ the location
of the optimal center of $P$ (namely, the centroid of $P$ in the case of the $k$-means
problem or the median of $P$ in the case of the $k$-median problem).
We will make us of the following lemma of~\cite{KSS10}.

\begin{lemma}[Lemma 5.3 in~\cite{KSS10}]
  \label{lem:sampling}
  Let $P$ be a set of points in $\R^d$ and $X$ be a random sample of size
  $O(\eps^{-3} \log (1/\eps))$ from $P$ and 
  $a$ and $b$ such that
  $a\le \cost(P, \opt_1(P)) \le b$.
  Then, we can construct a set $Y$ of $O(2^{1/\eps^{O(1)}} \log(b/\eps a))$
  points such that with constant probability there is
  at least one point $z \in X \cup Y$
  satisfying $\cost(P, \{z\}) \le (1 + 2\eps)\cost(P,\opt_1(P))$.
  Further, the time taken to construct $Y$ from $X$ is
  $O(2^{1/\eps^{O(1)}} \log(b/\eps a)d)$.
\end{lemma}

Our algorithm for obtaining a $(1+\eps)$-approximation
is as follows:
\begin{enumerate}
\item Compute a coreset $C$ for capacitated $k$-median
  as described by Lemma~\ref{lem:sampling}, and an estimate $\gamma$
  of the value of $\opt$ using the classic $O(\log n)$-approximation.

  In the remaining, we assume that the minimum pairwise distance
  between pairs of points of $C$ is at least $\eps \gamma/(n\log n)$
  since otherwise one can simply take a net of the input and the additive
  error is at most $\eps \opt$
  (see e.g.:~\cite{absFPTkmed}). Moreover, we assume that there is
  no cluster containing only one point of the coreset since these
  clusters can be ``guessed'' and dealt with separately.
  
\item Start with $\calC = \emptyset$, then 
  for each subset $S$ of $C$ of size $O(\eps^{-3} \log(k/\eps))$,
  for each $s = (1+\eps)^i$ in the interval $[\eps \gamma/(n \log n),
    \gamma]$ apply the procedure of Lemma~\ref{lem:sampling} with
  $a = s$ and $b$ = $(1+\eps)a$ and add the output of the procedure
  to $\calC$.
  We refer to $\calC$ as a set of approximate candidate centers.
\item Consider all subsets of size $k$ of $\calC$. For each 
  subset, compute the cost of using this set of centers for
  the capacitated $k$-median instance by using a min cost flow
  computation.
  Output the set of centers of minimum cost.
\end{enumerate}

We first discuss the running time of the algorithm.
The time for computing the coreset is polynomial by
Theorem~\ref{thm:coreset}.
Generating $\calC$ takes 
$|C|^{O(\eps^{-3} \log(1/\eps))} \cdot 2^{1/\eps^{O(1)}} \log((1+\eps)/\eps)d$
time.
For the last part, namely enumerating all subsets of $\calC$ of size $k$,
the running time is 
$|C|^{O(k\eps^{-3} \log(1/\eps))} \cdot 2^{k/\eps^{O(1)}} \log^k((1+\eps)/\eps)$.
Theorem~\ref{thm:coreset} implies that
$|C| = \text{poly}(k \log n\,\eps\inv)$ and so, the
algorithm has running time $(k\logn\,\e\inv)^{k\e^{-O(1)}}n^{O(1)}$. Again, the $(\logn)^{k\e^{-O(1)}}$ factor can be upper bounded by $(k/\e)^{k\e^{-O(1)}}$ or $n^{O(1)}$ based on whether or not $k\e^{-O(1)} < \f{\logn}{\log\logn}$, hence the improved running time in \Cref{thm:Euclidean:continous}.

We show that this algorithm provides a $(1+O(\eps))$-approximation.
Theorem~\ref{thm:coreset}
immediately implies that the solution found for the coreset
$C$ can be lifted to a solution for the original input at a cost of an
additive $O(\eps \opt)$.
For any (possibly weighted) set of client $A$ and set
of centers $B$, we define $\cost(A,B)$ to be the cost of
the best assignment of the clients in $A$ to the centers of $B$.

\begin{lemma}
  \label{lem:candidateset}
  The $\calC$ computed by the algorithm contains a set of
  centers $\tilde{S}$ that is such that $\cost(C, \tilde{S}) \le
  (1+\eps) \cost(C, \opt)$.
\end{lemma}
\begin{proof}
  This follows almost immediately from Lemma~\ref{lem:sampling}.
  By Lemma~\ref{lem:sampling},
  for each cluster $C_i^*$ of $\opt$, there exists a set $S_i^*
  \subseteq C_i^*$ of
  size at most $O(\eps^{-3} \log(k/\eps))$ such that applying the procedure
  of Lemma~\ref{lem:sampling} with the correct value of $a$
  to $S_i^*$ yields a set of points containing a point
  $z_i$ such that $\cost(C_i^*,z_i) \le (1+2\eps)\cost(C_i^* , \opt)$.
  Since the algorithm iterates over all subsets of size 
  $O(\eps^{-3} \log(k/\eps))$, and that the pairwise distance is at
  least $\eps \opt/n$, it follows that $S_i^*$ is one of the
  subset considered by the algorithm, and so $z_i$ is part of $\calC$.
\end{proof}

Finally, since the algorithm iterates over all subsets of $\calC$ of
size at most $k$, Lemma~\ref{lem:candidateset} implies that 
there exists a set $\{z_1,\ldots,z_k\}$ that is considered by the
algorithm and on which solving a min cost flow instance
yields a solution of cost at most $(1+O(\eps)) \cost(\calP,\opt)$.

\subsection{The Non-Uniform Case --
  Proof of Theorem~\ref{thm:Euclidean:discrete}}
We now consider the non-uniform case. In this setting, the input
consists of a set of points in $\R^d$ together with a set of
candidate centers in $\R^d$ and a capacity $\eta_f$
for each such candidate
center.
We make use of the following lemma. As slightly worse
bound for the lemma can also be found in~\cite{Mahabadi:2018}.

\begin{lemma}[\cite{abs-1810-09250}]
  \label{lem:dimreduction}
  Let $\eps \in (0,1)$ and $X \subseteq \R^d$
  be arbitrary with $X$ having size $n>1$.
  There exists $f: \R^d \mapsto \R^m$ with $m=O(\eps^{-2}\log n)$
  such that
  $\forall x \in X$,  $\forall y\in \R^d$,
  $||x-y||_2 \le ||f(x)-f(y)||_2 \le (1+\eps)||x-y||_2$.
\end{lemma}

We describe a polynomial-time approximation scheme. Let $\eps>0$.
The algorithm
is as follows. The first step of the algorithm is identical to
the continous case.

\begin{enumerate}
\item Compute a coreset $C$ for capacitated $k$-median
  as described by Theorem~\ref{lem:sampling}, and an estimate $\gamma$
  of the value of $\opt$ using the classic $O(\log n)$-approximation.

  In the remaining, we assume that the minimum pairwise distance
  between pairs of points of $C$ is at least $\eps \gamma/(n\log n)$
  since otherwise one can simply take a net of the input and the additive
  error is at most $\eps \opt$
  (see e.g.:~\cite{absFPTkmed}). Moreover, we assume that there is
  no cluster containing only one point of the coreset since these
  clusters can be ``guessed'' and dealt with separately.
\item Apply Lemma~\ref{lem:dimreduction} to the points of the coreset
  to obtain a set of points in a
  Euclidean space of dimension 
  $\frac{\log k + \log \log n}{\eps^{O(1)}}$. Let $C^*$ and $A^*$ be
  respectively the image of the coreset points and of the candidate
  centers through the projection.
  
\item Start with $\mathcal{V} = \emptyset$
  For each point $p$ of the coreset do the following:
  For each $i \in \{1,2,\ldots,n^2\}$, consider the
  $i$th-\emph{ring} defined by
  $\ball(p,(1+\eps)^i \eps \gamma/(n\log n))
  \setminus \ball(p,(1+\eps)^{i-1} \eps \gamma/(n\log n))$
  and choose an $\eps \cdot (1+\eps)^i \eps \gamma/(n\log n)$-net.
  Consider the Voronoi diagram induced by the points of the net. Then,
  for each Voronoi cell, add to $\mathcal{V}$
  the $k$ candidate centers of $A^*$ in the cell that
  are of maximum capacity.
\item Enumerate all possible subset of $\mathcal{V}$ of size $k$
  and output the one that leads to the solution
  of minimum cost.
\end{enumerate}

\paragraph{Correctness.}
Theorem~\ref{thm:coreset} implies that finding a near-optimal solution
for the coreset points yields a near-optimal solution for the input point
set.

Lemma~\ref{lem:dimreduction} immediately implies that, given the coreset
construction $C$, and the projection of the coreset points onto a
$\frac{\log k + \log \log n}{\eps^{O(1)}}$-dimensional Euclidean space,
finding a near-optimal set of centers in $A^*$ yields a near-optimal
set of centers in $A$ through the inverse of the projection.

Therefore, it remains to show that the set $\mathcal{V}$ contains
a set of candidate centers that yields a near-optimal solution. To see
this, consider each center of the optimal solution in $A^*$. For each
such optimal center $f$, consider the closest coreset point $c(f)$
together with the ring of $c(f)$ containing $f$. Let $j$ be the index
of this ring, namely $f \in \ball(p,(1+\eps)^j \eps \gamma/(n\log n))
\setminus \ball(p,(1+\eps)^{j-1} \eps \gamma/(n\log n))$.

By definition of the net, there exists a point $p$ of the net at distance
at most $\eps \cdot \ball(p,(1+\eps)^j \eps \gamma/(n\log n))
\le 2\eps ||c-c(f)||_2$ from $c(f)$. Therefore, consider the Voronoi
cell of $p$ and the top-$k$ candidate centers in terms of capacity.
If $f$ is part of this top-$k$, then $f$ is part of $\mathcal{V}$ and
we are done. Otherwise, it is possible to associate to $f$ a center $f^*$
that
has capacity at least the capacity of $f$, and so for all the optimal
centers simultaneously since we consider the top-$k$.
Therefore, consider replacing $f$ by $f^*$ in the optimal solution.
The change in cost is at most, by the triangle inequality,
$4\eps ||c-c(f)||_2$ since both centers are in the Voronoi cell of
$p$. Finally, since $c$ is the closest client to $c(f)$, the cost
increases by a factor at most $(1+4\eps)$ for each client and
the correctness follows.

\paragraph{Running time.}
We now bound the running time. The first two steps are clearly
polynomial time. An
$\eps \cdot (1+\eps)^i \eps \gamma/(n\log n)$-net of
a ball of radius $(1+\eps)^i \eps \gamma/(n\log n)$ has size
$\eps^{-O(d)}$ and so in this context, after Step 2, a
size $\eps^{-(\frac{\log k + \log \log n}{\eps^{O(1)}})}$.
Since for each element of the net, $k$ centers are chosen and since
the number of rings is, by Step 1, at most $O(\eps^{-2} \log n)$,
the total size of $\mathcal{V}$ is at most
$|C| k \eps^{-2} \log n \eps^{-(\frac{\log k + \log \log n}{\eps^{O(1)}})}$
which is at most
$|C| \eps^{-2} (k\log n)^{\eps^{-O(1)}} = (k\eps^{-1}\log n )^{\eps^{-O(1)}}$.
Enumerating all subsets of size $k$ takes time
$(k\eps^{-1}\log n )^{k\eps^{-O(1)}}$ and the theorem follows.

\bibliographystyle{abbrv}
\bibliography{clustering}
\appendix
\section{Proof of \Cref{lem:step2-2}} \label{sec:lemma20}

Our strategy is to construct a flow for $\flowi(X)$ with cost $\le g(\E[X])+0.49\e NR = g(\1)+0.49\e NR$, and ensure that it succeeds with probability $\ge1-n^{-10}$. Essentially, we want to look at the min-cost flow of $\flowi(\1)$, and construct a flow that is competitive to it.

First, for all clients not in $C'$, we route their $1$ unit of demand in the same manner as in the min-cost flow of $\flowi(\1)$. Therefore, it remains to assign the demands in $C'$, as well as the extra $N-\sum_{c\in C'}d_c$ demand at ring center $f'$.

Consider the optimal min-cost flow on $\flowi(\1)$, which we can assume is an integral flow. For each center $f\in F$, define $C'_f\s C'$ as the clients served by $f$ in this flow. For each center $f$, we now focus on routing the demands of the sampled clients in $C'_f$. Let $S_f\s C'_f$ be the sampled clients in $C'_f$, i.e., the clients whose coordinates in $X$ are nonzero. Note that $S_f$ is distributed as $\Bern(|C'_f|,r/N)$. If $S_f \cd r/N \le |C'_f|$ (i.e., we under-sampled), then in $\flowi(X)$, we will route:
\BE
\im[(1)] $N/r$ units of flow from each center in $S_f$ to $f$, and
\im[(2)] $|C'_f| - |S_f|\cd N/r$ units of flow from ring center $f'$ to $f$.
\EE
Otherwise, if $S_f\cd r/N>|C'_f|$ (i.e., we over-sampled), then we will first randomly sub-sample from $S_f$ to obtain a set $S'_f$ of size exactly $\lf |C'_f| \cd N/r\rf$, and then apply steps (1) and (2) on $S'_f$ instead of $S_f$, and finally route the remaining $N/r$ units of flow from each client in $S_f\setminus S'_f$ to the ring center $f'$. It is easy to verify that if we do this for each $C'_f$, we obtain a feasible flow for $\flowi(X)$.

While the flow construction is the same for each $C'_f$, we will case on whether $\E[|S_f|]=|C'_f|\cd r/N$ is above or below $0.05\e r/k$ in our analysis.

\para{Case $\E[|S_f|]\ge0.05\e r/k$.}
First, suppose that $\E[|S_f|]\ge0.05\e r/k \iff |C'_f| \ge 0.05\e N/k$, considered to be the ``hard'' case. We begin with a concentration bound on the size of $S_f$:
\BSCL\label{clm:S-conc}
With probability $\ge 1-n^{-20}$,
\[ \left| |S_f| - |C'_f| \cd \f rN \right|\le 0.1\e\cd |C'_f|\cd\f rN .\]
\ESCL
\BP
We apply the Chernoff bound
\begin{align*}
 \Pr\big[\big||S_f|-|C'_f|\cd r/N\big|\le 0.1\e \cd |C'_f|\cd r/N\big] &\le \exp(-\Th(\e^2 \cd |C'_f| \cd r/N)) \\&\le \exp(-\Th(\e^2 \cd 0.05\e r/k)) \\&\le\exp(-20\ln n) = n^{-20}
\end{align*}
 for large enough $\g$ in the definition of $r$.
\EP

If $S_f > |C'_f|\cd N/r$, then let $S'_f$ be the set $S_f$ with $\big\lc S_f - |C'_f| \big\rc$ random clients removed; otherwise, let $S'_f = S_f$. (So essentially, we are merging the under-sampled and over-sampled cases together.) By construction, $|S'_f|\le |C'_f|\cd r/N$, and by \Cref{clm:S-conc}, with probability $\ge 1-n^{-20}$,
\begin{gather}
|S_f| - |S'_f| \le 0.1\e \cd |C'_f| \cd \f rN + 1 \le 0.11\e \cd |C'_f| \cd \f rN,  \label{eq:2}
\end{gather}
that is, at most $0.11\e\cd|C'_f|\cd r/N$ many clients were removed. Also, by \Cref{clm:S-conc},
\begin{gather}
|C'_f|\cd\f rN - |S'_f| \le 0.1\e \cd |C'_f| \cd \f rN  \label{eq:3}
\end{gather}

For each client $c\in S'_f$, we route $N/r$ units of flow from $c$ to $f$, costing $N/r \cd d(c,f)$. To bound the cost of the flow, we use the following concentration bound from [Chen09]:




\BL[Lemma 3.2 from Chen09]\label{lem:3.2}
Let $M\ge0$ and $\eta$ be fixed constants, and let $h(\cd)$ be a function defined on a set $V$ such that $\eta\le h(p)\le \eta+M$ for all $p\in V$. Let $U=\{p_1,\lds,p_s\}$ be a set of $s$ samples drawn independently and uniformly from $V$, and let $\de>0$ be a parameter. If $s\ge (M^2/2\de^2)\ln(2/\la)$, then $\Pr\big[\big|\f{h(V)}{|V|}-\f{h(U)}{|U|}\big|\ge\de\big]\le\la$, where $h(U)=\sum_{u\in U}h(u)$ and $h(V)=\sum_{v\in V}h(v)$.
\EL

Fix an arbitrary integer $s\in [1-0.1\e,1]\cd |C'_f|\cd r/N$, and for now, condition on the event that $|S'_f|=s$. Observe that under this conditioning, the clients in $S'_f$ is a uniformly random sample of $s$ clients in $C'_f$. We apply this lemma with $M:=2R$, $h(c):=d(c,f)$, $U:=S'_f$, $V:=C'_f$, $\de:=0.1\e R$, $\la:=n^{-20}$. To verify the requirements of the lemma, observe that
\begin{align*}
 s &\ge (1-0.1\e)\cd |C'_f|\cd r/N \ge (1-0.1\e) \cd 0.05\e N/k \cd r/N = \Th(\e r/k) = \Th(k \logn/\e^2) \\&\ge (M^2/2\de^2) \ln (2/\la) 
\end{align*}
for sufficiently large constant $\gamma$ in $r=\g k\logn/\e^3$.
Also, since all clients in $C'$ are within some ball of radius $R$, by the triangle inequality, all values of $h(c)=d(c,f)$ are contained in an interval of length $2R$, so the function $h$ satisfies the condition of the lemma for some $\eta$. Therefore, we invoke \Cref{lem:3.2} to obtain that with probability $\ge 1-n^{-20}$,
\[
\left| \f{h(V)}{|V|} - \f{h(U)}{|U|}\right|
=
\left| \f{h(C'_f)}{|C'_f|} - \f{h(S'_f)}{|S'_f|}\right|
\le\de \iff
\left| \f{h(V) \cd |S'_f|}{|C'_f|} - h(U)\right| \le \de|S'_f|\]
\begin{align}
\implies h(S'_f) &\le h(C'_f)\cd \f{|S'_f|}{|C'_f|} + \de|S'_f|   \nonumber
\\&\le h(C'_f) \cd \f{|S'_f|}{|C'_f|} +\de  \lp|C'_f| \cd \f rN\rp   \nonumber
\\ \implies h(S'_f) \cd \f Nr &\le h(C'_f) \cd \f{|S'_f|}{|C'_f|} \cd \f Nr + \de |C'_f|  \nonumber
\\&=  h(C'_f)\cd \f{|S'_f|}{|C'_f|} \cd\f Nr + 0.1\e R |C'_f|   .   \label{eq:flow1}
\end{align}
Observe that $h(S'_f) \cd N/r$ is exactly the cost of our flow in step (1), and $h(C'_f)$ is exactly the cost of the flow paths to $f$ in the min-cost flow of $\flowi(\1)$. Finally, we take a union bound over all values of $s$.

Next, we bound the routing cost of step (2), namely, routing $|C'_f|-|S'_f|\cd N/r$ flow from ring center $f'$ to $f$. By the triangle inequality, for all clients $c\in C'_f\s \ball(f',R)$, we have $d(f',f)\le d(c,f)+R=h(c)+R$. In particular, averaging gives $d(f',f)\le h(C'_f)/|C'_f|+R$. Therefore, the total cost of step (2) is
\begin{align}
\lp |C'_f|-|S'_f|\cd \f Nr\rp \cd d(f',f) &\le \lp |C'_f|-|S'_f|\cd \f Nr\rp \cd \lp \f{h(C'_f)}{|C'_f|}+R \rp   \nonumber
\\&\le h(C'_f) - h(C'_f)\cd\f{|S'_f|}{|C'_f|}\cd \f Nr + \lp|C'_f|-|S'_f|\cd \f N r \rp \cd R   \nonumber
\\&\stackrel{(\ref{eq:3})}\le h(C'_f) - h(C'_f)\cd\f{|S'_f|}{|C'_f|}\cd \f Nr + 0.1\e\cd |C'_f| \cd R   \label{eq:flow2}
\end{align}

Finally, if we over-sampled the set $S'_f$, then we need to route $N/r$ flow from each client in $S_f\setminus S'_f$ to the ring center $f'$. Since all clients in $C'_f$ are distance $\le R$ from $f'$, this costs at most
\begin{gather}
(|S_f| - |S'_f|) \cd \f Nr \cd R \stackrel{(\ref{eq:2})}\le \lp0.11\e\cd|C'_f|\cd\f rN\rp \cd \f Nr\cd R \le 0.11\e \cd |C'_f| \cd R.  \label{eq:flow3}
\end{gather}

The total flow is the sum of (\ref{eq:flow1}), (\ref{eq:flow2}), and (\ref{eq:flow3}), which is at most 
\[ h(C'_f) + 0.31\e R|C'_f| = \sum_{c\in C'_f} d(c,f) + 0.31\e R|C'_f| .\]

\para{Case $\E[|S_f|]\le 0.05\e r/k$.}
Now suppose that $\E[|S_f|]\le0.05\e r/k \iff |C'_f| \le 0.05\e N/k$. In this case, we have a corresponding concentration bound, similar to \Cref{clm:S-conc}, but this one is one-sided and more generous:
\BSCL\label{clm:S-conc-small}
With probability $\ge 1-n^{-20}$,
\[ |S_f|\le 0.06\e r/k .\]
\ESCL
\BP
We apply the Chernoff bound 
\[ \Pr[|S_f|\ge \E[|S_f|] + 0.01\e r/k] \le \exp(-\Th(\e r/k))\le\exp(-20\ln n)=n^{-20} \]
for large enough $\g$ in the definition of $r$.
\EP
We now compute the costs of the flows, analogous to the bounds of  (\ref{eq:flow1}), (\ref{eq:flow2}), and (\ref{eq:flow3}) in the other case; here, our bounds will be a lot looser. For (\ref{eq:flow1}), we instead have
\begin{align*}
\sum_{c\in S'_f}d(c,f)\cd\f Nr & \le \sum_{c\in S'_f}(d(c,f')+d(f',f))\cd \f Nr
\\&\le \sum_{c\in S'_f}(R+d(f',f))\cd \f Nr
\\& = |S'_f|\cd(R+d(f',f))\cd\f Nr  ;
\end{align*}
for (\ref{eq:flow2}), we instead have
\begin{align*}
\lp |C'_f|-|S'_f|\cd \f Nr\rp \cd d(f',f) &= \sum_{c\in C'_f} d(f',f) - |S'_f|\cd\f Nr\cd d(f',f)
\\&\le \sum_{c\in C'_f}(d(f',c)+d(c,f))-|S'_f|\cd\f Nr\cd d(f',f)
\\&\le \sum_{c\in C'_f}(R+d(c,f))-|S'_f|\cd\f Nr\cd d(f',f)
\\&= |C'_f|\cd R+ \sum_{c\in C'_f}d(c,f)  - |S'_f|\cd\f Nr\cd d(f',f)  ;
\end{align*}
and for (\ref{eq:flow3}), we instead have
\begin{align*}
(|S_f| - |S'_f|) \cd \f Nr \cd R \le |S_f| \cd \f Nr \cd R   .
\end{align*}
Altogether, summing up these bounds gives a total flow at most
\begin{align*}
 &\sum_{c\in C'_f}d(c,f)+|S'_f|\cd R\cd \f Nr + |C'_f| \cd R + |S_f| \cd \f Nr\cd R
\\&\le \sum_{c\in C'_f}d(c,f)+ \f{0.06\e r}k \cd R\cd \f Nr + \f{0.05\e N}k\cd R + \f{0.06\e r}k \cd \f Nr\cd R
\\& = \sum_{c\in C'_f}d(c,f)+ \f{0.17\e NR}k .
\end{align*}

\para{Putting things together.}
Regardless of case, the flow is upper bounded by \[ \sum_{c\in C'_f} d(c,f) + 0.31\e R|C'_f| +\f{0.17\e NR}k  .\]
Finally, summing over all $f\in F$, we obtain a total cost of at most
\[ \sum_{f\in F} \lp \sum_{c\in C'_f} d(c,f) + 0.31\e R|C'_f|+\f{0.17\e NR}k\rp = \sum_{f\in F}\sum_{c\in C'_f} d(c,f) + 0.31\e R|C'|+0.17\e NR \]
for routing all demands in $C'$, as well as the extra demand at ring center $f'$. That is, our flow is only $0.48\e NR $ additively worse off at routing these demands, compared to the min-cost flow of $\flowi(\1)$. Lastly, since the two flows coincide outside of these demands, \Cref{lem:step2-2} follows.






\section{Multiple rings case} \label{sec:multiple-rings}
Consider now the case of multiple rings. Our idea is to consider a single martingale across all rings, but \emph{bound the concentration for each ring separately} (and then take a union bound over all rings). 
We treat each ring separately because the rings have varying scales, and martingale inequalities are nicest when all random variables are on the same scale.

In particular, the rings in the martingale are ordered sequentially, so that all vertices in the first ring appear first, then all vertices in the second ring, etc. The ordering can be arbitrary, but for simplicity of analysis, let us assume the ordering is lexicographic by $(i,R)$, with the $<$ comparison operator on $(i,R)$ indicating lexicographic comparison.

We will now define a function $g_{i,R}$ for each ring $C_{i,R}$, similar to the function $g$ from \Cref{sec:single-ring}. Fix a ring $C_{i,R}$; the function $g_{i,R}$ will depend on the samples $S_{i',R'}$ of all the other rings, i.e., $S_{i',R'}$ for all $(i',R')\ne(i,R)$. So suppose we fix the sample $S_{i',R'}$ for all other rings. For input a vector $\bd\in\R_+^{C_{i,R}}$ (indexed by clients in $C_{i,R}$), consider a min-cost flow instance $\flowi_{i,R}(\bd)$ on the graph metric with the following demands. First, set demand $d_c$ at each client $c\in C_{i,R}$ and demand $|C_{i,R}|-\sum_{c\in C_{i,R}}d_c$, as in $\flowi(\bd)$ in \Cref{sec:single-ring}. Now, for each sample $C_{i',R'}$, $(i',R')\ne(i,R)$, set demand $|C_{i',R'}|/r$ at each client $c\in S_{i',R'}$; note that $|S_{i',R'}|\cd |C_{i',R'}|/r = |C_{i',R'}|$, so no additional demands needs to be set at ring center $f'_{i'}$. Observe that  $\flowi_{i,R}(\bd)$  is a feasible min-cost flow instance.

The following is an analogue of \Cref{clm:lip} from \Cref{sec:single-ring}, and the proof is identical:
\BCL
For any ring $C_{i,R}$ and fixed samples $S_{i',R'}$ for all $(i',R')\ne(i,R)$, the function $g_{i,R}(\bd)$ is $R$-Lipschitz w.r.t.\ the $\el_1$ distance in $\R_+^{C_{i,R}}$.
\ECL

Recall that the average of a set of $R$-Lipschitz functions is also $R$-Lipschitz. Suppose now, we only fix the samples $S_{i',R'}$ for $(i',R')<(i,R)$ and instead take the expectation over $(i',R')>(i,R)$. Then, the expected value of $g_{i,R}(\bd)$, that is,
\begin{gather}
  \E_{S_{i',R'}:(i',R')>(i,R)} \lb g_{i,R} (\bd) \big\vert S_{i',R'}:(i',R')<(i,R) \rb ,   \label{eq:expected}
\end{gather}
is also $R$-Lipschitz.
From now on, we will abbreviate (\ref{eq:expected}) with simply $\E[g_{i,R}(\bd)]$, i.e., the expectation and conditioning are implicit. We therefore obtain:
\BCL
For any ring $C_{i,R}$, $\E[g_{i,R}(\bd)]$ is $R$-Lipschitz w.r.t.\ the $\el_1$ distance in $\R_+^{C_{i,R}}$.
\ECL

We now focus on the analogues of \Cref{lem:step1,lem:step2}, whose proofs follow through when $g$ is replaced by $g_{i,R}$. The reasoning is that our arguments are not affected by the nature of the demands outside the ones attributed to $C_{i,R}$. For example, in the proof of \Cref{lem:step2-2}, we route demands outside of those attributed to $C_{i,R}$ in the same manner as the min-cost flow in $\flowi(\1)$, and this argument does not depend on these particular demands.

\BL\label{lem:step1-m}
Assume that $|C_{i,R}| > \Th(k\log n/\e^3)$. With probability $\ge 1-n^{-(k+20)}$, we have $| g_{i,R}(X) - \E[g_{i,R}(X)] | \le \e|C_{i,R}|R/2$.
\EL
\BL\label{lem:step2-m}
Assume that $|C_{i,R}| > \Th(k\log n/\e^3)$. Then, $| \E[g_{i,R}(X)] - g_{i,R}(\E[X]) | \le \e|C_{i,R}|R/2$.
\EL

The remaining arguments in \Cref{sec:single-ring}, the ones that link \Cref{lem:step1,lem:step2} to \Cref{lem:single-ring}, also follow through. In particular, instead of \Cref{lem:single-ring}, we are now bounding the deviation of the random variable
\[ \E_{>(i,R)}\big[\fkmed(W',F)\big] := \E_{S_{i',R'}:(i',R')>(i,R)} \lb \fkmed(W',F) \big\vert  S_{i',R'}:(i',R')<(i,R) \rb \]
from the constant
\[
\E_{\ge(i,R)}\big[\fkmed(W',F)\big] := \E_{S_{i',R'}:(i',R')\ge(i,R)} \lb \fkmed(W',F) \big\vert  S_{i',R'}:(i',R')<(i,R) \rb .
\]
With these definitions in mind, \Cref{lem:step1-m,lem:step2-m} together imply:

\BL\label{lem:multiple-ring}
W.h.p., for any set of $k$ centers $F\s\F$ satisfying $\kmed(C,F)<\infty$,
\begin{gather*}
 \left|\E_{>(i,R)}\big[\fkmed(W',F)\big] -\E_{\ge(i,R)}\big[\fkmed(W',F)\big]  \right| \le \e |C_{i,R}|R . 
\end{gather*}
\EL

Recall that there are only $O(k\log n)$ rings. Consider the process of going through the rings $C_{i,R}$ in order, and for each one, apply \Cref{lem:multiple-ring} conditioned on the choices of $S_{i',R'}$ for $(i',R')<(i,R)$. W.h.p., for any set of $k$ centers $F\s \F$, the total deviation is at most $\sum_{(i,R)} \e |C_{i,R}|R$, that is,
\begin{gather}
\left|\fkmed(W',F) - \E\big[\fkmed(W',F)\big] \right| \le \sum_{(i,R)}\e|C_{i,R}|R.\label{eq:mcf-dev}
\end{gather}

It remains to bound $\sum_{(i,R)}\e|C_{i,R}|R$ by $OPT$. Observe that the bicriteria solution (line~\ref{line:1}) has solution $O(OPT)$, and for each ring $C_{i,R}$, pays $\sum_{c\in C_{i,R}}d(c,f_i') \ge |C_{i,R}|\cd R/2$, using that $c\notin \ball(f'_i,R/2)$. Therefore, continuing (\ref{eq:mcf-dev}), we get
\begin{align*}
\left|\fkmed(W',F) - \E\big[\fkmed(W',F)\big] \right| &\le \sum_{(i,R)}\e|C_{i,R}|R
\\&= 2\e \cd \sum_{(i,R)}|C_{i,R}|\cd R/2
\\&\le 2\e\cd ALG' \le O(\e) \cd OPT .
\end{align*}
Finally, scaling down $\e$ by a constant factor, so that $O(\e)\cd OPT$ becomes $\e\cd OPT$, produces the desired coreset, finishing \Cref{thm:coreset}.

\section{Coreset for $k$-means}

The main technical difference in the $k$-means case is that we can no longer view the problem as a formulation of $\flowi$ from \Cref{sec:single-ring}, even if we square the distances. This is because $\flowi$ does not stop a flow path from using multiple short edges which, with each edge length squared, is a lot smaller in total than the entire distance squared. Therefore, we instead consider a bipartite matching formulation, where clients, which have positive demand, are on the left (call them \emph{sources}), and centers, which have negative demand, are on the right (call them \emph{sinks}). 
As for the potential demand $N-\sum_{c\in C'}d_c$ at ring center $f'=f'_i$, if this demand is positive, then $f'$ is one the left (it becomes a source); if it is negative, then $f'$ is on the right (it becomes a sink). We will not give this modified instance a name, since the rest of this section follows closely to \Cref{sec:single-ring}. 

\subsection{Single-ring case ($k$-means)}

For our single-ring case of the $k$-means coreset, we have the following corresponding lemma. Observe that the bound is slightly worse: there is an extra $O(\e)\cd OPT$ term in the error. This is not morally a problem, since there are only $O(k\logn)$ rings, but it will make our overall core-set bounds worse.

\BL\label{lem:single-ring-means}
W.h.p., for any set of $k$ centers $F\s\F$ satisfying $\kmeans(C,F)<\infty$,
\begin{gather}
 | \fkmeans(W',F) -  \kmeans(C,F) | \le \e NR + O(\e)\cd OPT. \label{eq:1}
\end{gather}
\EL

The proof of \Cref{lem:single-ring-means} follows the outline of \Cref{sec:single-ring}. The rest of this section describes the modifications we make compared to \Cref{sec:single-ring}.

\para{Step (1).}
We apply \Cref{lem:3.2} with a different set of parameters. This time, define $h(c):=d(c,f)^2$ to capture $k$-means distance. Since all distances $d(c,f)$ for $f\in C'_f$ are within an interval of length $2R$, the average distance $\sum_{c\in C'_f}d(c,f)/|C'_f|$ also belongs to this interval. Therefore, for any client $c\in C'_f$,
\begin{align}
0 \le h(c) = d(c,f)^2 &\le \lp \f1{|C'_f|}\sum_{c\in C'_f}d(c,f) + 2R\rp^2   \nonumber
\\&\le 2\lp\f1{|C'_f|}\sum_{c\in C'_f}d(c,f) \rp^2 + 2\big(2R\big)^2   \nonumber
\\&\le \f2{|C'_f|}\sum_{c\in C'_f}d(c,f)^2 + 8R^2  ,   \label{eq:means-1}
\end{align}
where the last inequality follows from $\sum_{c\in C'_f}d(c,f)^2 \ge |C'_f|\big(\f1{|C'_f|} \sum_{c\in C'_f}d(c,f)\big)^2$, which is an application of Cauchy-Schwarz.

We define $M$ as the expression (\ref{eq:means-1}). Then, apply \Cref{lem:3.2} with remaining assignments $U:=S'_f$, $V:=C'_f$, $\de:=0.1\e M$, $\la:=n^{-20}$, which satisfy the requirements of the lemma. Following (\ref{eq:flow1}), we can bound the flow in step (1) as
\begin{align*} 
h(S'_f) \cd \f Nr &\le h(C'_f) \cd \f{|S'_f|}{|C'_f|} \cd \f Nr + \de |C'_f|
\\&\stackrel{\mathclap{\text{Clm}\,\ref{clm:S-conc}}}\le h(C'_f) \cd \f{(1+0.1\e)|C'_f|\cd r/N}{|C'_f|} \cd \f Nr + \de |C'_f|
\\&= (1+0.1\e)\cd h(C'_f)+\de|C'_f|
\\&= (1+0.1\e)\cd h(C'_f)+ 0.1\e \lp \f2{|C'_f|}\sum_{c\in C'_f}d(c,f)^2 + 8R^2   \rp  |C'_f|
\\&= (1+0.1\e)\cd h(C'_f) + 0.2\e\sum_{c\in C'_f}d(c,f)^2 + 0.8\e R^2|C_f|
\\&= (1+0.3\e)\sum_{c\in C'_f}d(c,f)^2  + 0.8\e R^2|C_f| .
\end{align*}

\para{Step (2) and oversampling.}
In the $k$-means case, we cannot directly bound the cost through (\ref{eq:flow2})~and~(\ref{eq:flow3}), and need to split into two cases. First, if we under-sampled, then $f'$ is a source,  so we still need to route $N - |S_f'|\cd N/r$ flow from $f'$ to $f$, which as cost
\begin{align*}
\lp N-|S_f'|\cd \f Nr \rp \cd d(f',f) &\stackrel{(\ref{eq:means-1})}\le \lp N-|S_f'|\cd \f Nr \rp \cd  \lp \f2{|C'_f|}\sum_{c\in C'_f}d(c,f)^2 + 8R^2\rp 
\\&\stackrel{\mathclap{\text{Clm}\,\ref{clm:S-conc}}}\le \big( 0.1\e\cd |C'_f|\big) \cd \lp \f2{|C'_f|}\sum_{c\in C'_f}d(c,f)^2 + 8R^2\rp 
\\&= 0.2\e\sum_{c\in C'_f}d(c,f)^2 + 0.8\e|C'_f|R^2.
\end{align*}


Otherwise, if we over-sampled, then we still need to route at most $N/r$ flow from one client to $f$ (since $|S'_f|=\lf |C'_f| \cd N/r\rf$, not $|C'_f| \cd N/r$), as well as the remaining flow from clients in $S_f\setminus S'_f$ to $f'$. The first can easily be bounded by
\[
0.01\e\sum_{c\in C'_f}d(c,f)^2 + 0.01\e|C'_f|R^2 ,
\]
and the second by
\begin{align*}
\sum_{c\in S_f\setminus S'_f}d(c,f')^2 \cd \f Nr &\le |S_f\setminus S'_f| \cd R^2 \cd \f Nr
\\&\stackrel{(\ref{eq:2})}\le \lp0.11\e\cd|C'_f|\cd\f rN\rp \cd R^2\cd \f Nr
\\&= 0.11|C'_f|R^2.
\end{align*}
Altogether, summing up the expression gives $(1+O(\e))\sum_{c\in C'_f}d(c,f)^2+0.8\e|C'_f|R^2$.

\para{Case $\E[|S_f|]\le0.05\e r/k$.}
In this case, since the total demand is so small, we can be a lot looser. For brevity, we only bound one of the flows and omit the rest. Namely, the cost to route all flow in $S'_f$ to $f'$ is at most
\begin{align*}
\sum_{c\in S'_f}d(c,f')\cd\f Nr &\stackrel{(\ref{eq:means-1})}\le |S'_f|\cd  \lp \f2{|C'_f|}\sum_{c\in C'_f}d(c,f)^2 + 8R^2\rp\cd\f Nr
\\&\stackrel{\mathclap{\text{Clm}\,\ref{clm:S-conc-small}}}\le \ \ \f{0.06\e r}k\cd  \lp \f2{|C'_f|}\sum_{c\in C'_f}d(c,f)^2 + 8R^2\rp\cd\f Nr
\\&=\f{0.12\e}k\sum_{c\in C'_f}d(c,f)^2 + \f{0.48\e NR^2}k.
\end{align*}
Adding up the costs, we get  $(1+O(\e))\sum_{c\in C'_f}d(c,f)^2  + O(\e)R^2|C'_f|$. Finally, summing over all $f\in F$, we obtain a total cost of
\begin{align*}
\sum_{f\in F} \lb (1+O(\e))\sum_{c\in C'_f}d(c,f)^2  + O(\e)R^2|C'_f| \rb &= (1+O(\e))\sum_{f\in F}\sum_{c\in C'_f}d(c,f)^2 + O(\e)NR^2.
\end{align*}

Now observe that $\sum_{f\in F}\sum_{c\in C'_f}d(c,f)^2 \le OPT$, since it describes part of the min-cost flow of $\flowi(\1)$, which has cost exactly $OPT$. Therefore, the total error is at most $O(\e)\cd OPT + O(\e) NR^2$, proving \Cref{lem:single-ring-means}.

\subsection{Multiple-rings case ($k$-means)}

For the multiple-rings case, the arguments are identical, and we sum the error in \Cref{lem:single-ring-means} over each ring $C_{i,R}$. This totals 
\[ \sum_{(i,R)} O(\e)|C_{i,R}|R + \sum_{(i,R)}O(\e)\cd OPT \le O(\e)\cd OPT + O(k\logn)\cd O(\e)\cd OPT .\]

Recall that the core-set has size $O(k\logn/\e^3)$. To obtain error $\e\cd OPT$, we need to reset $\e \gets \Th(\e/(k\logn))$, so the total core-set size becomes $O(k^4\log^4n/\e^3)$.

\section{Miscellaneous proofs}

\subsection{Proof of \Cref{thm:m2}}\label{sec:m2}

We begin with a statement on a normalized version of \Cref{thm:m2}:
\BT\label{thm:m1}
Let $x_1,\lds,x_n$ be independent random variables taking value $1$ with probability $p$ and value $0$ with probability $1-p$, and let $g:[0,1]^n\to \R$ be a $1$-Lipschitz function in $\el_1$ norm. Define $X:=(x_1,\lds,x_n)$ and $\mu:=\E[g(X)]$. Then, for $0\le\e\le1$:
\[ \Pr\big[\big|g(X)-\E[g(X)] \big| \ge \e pn\big] \le 2e^{-\e^2pn/3} \]
\ET

We follow Appendix A.1 of \cite{AlS16}, which proves the regular multiplicative Chernoff bound, and adapt it for martingales.

Consider the following $n$-step martingale. We start off with the (constant) random variable $\E[g(X)]$. On step $i\in[n]$, we reveal the choice of $x_i$, and our random variable becomes $\E[g(X)]$ conditioned on the outcomes of $x_1,\lds,x_i$, that is, $\E\lb g(X) \big\vert x_1,\lds,x_i \rb$. Let $\De_i := \E\lb g(X) \big\vert x_1,\lds,x_i \rb - \E\lb g(X) \big\vert x_1,\lds,x_{i-1} \rb$ be the change in our random variable on step $i$. Therefore, our random variable after step $i$ is $\E[g(X)]+(\De_1+\cds+\De_i)$, and at the end of the process, our deviation from $\E[g(X)]$ is exactly $\sum_{i=1}^n\De_i$. Therefore,
\[ \Pr [ | g(X)-\E[g(X)]| \ge \e pn ] = \Pr\big[ \big|\textstyle\sum_{i=1}^n \De_i \big| \ge \e pn \big],\]
so we focus on bounding $|\sum_i\De_i|$ instead.

Since $g$ is $1$-Lipschitz, a standard argument in martingale analysis says that $\De_i \in [-1,1]$ with probability $1$. Furthermore, since $\E[\De_i]=0$, there must exist $z_i\in[0,1]$ (depending on $x_1,\lds,x_{i-1}$) such that $\De_i$ takes value $(1-p)z_i$ with probability $p$ and $-pz_i$ with probability $1-p$. We now state the following helper claim:

\BCL\label{clm:z}
Fix $z\in[0,1]$, and let $\De$ be a random variable that is  $(1-p)z$ with probability $p$ and $-pz$ with probability $1-p$. Then, for any $\la\in\R$,
\[ \E\lb e^{-\la \De}\rb \le e^{-\la p}\lb pe^\la+(1-p) \rb .\]
\ECL
\BP
We have
\begin{gather}
\E\lb e^{-\la\De}\rb=pe^{\la(1-p)z}+(1-p)e^{-\la pz}. \label{eq:a1}
\end{gather}
We just need to show that this expression is maximized when $z=1$, so that $\E[e^{-\la\De}]\le pe^{\la(1-p)}+(1-p)e^{-\la p} = e^{-\la p}[ pe^\la+(1-p) ]$. Differentiating (\ref{eq:a1}) w.r.t.\ $z$, we get $\la p(1-p) e^{-\la pz} [ e^\la-1]$, which is nonnegative because $\la[e^\la-1]\ge0$ for all $\la$, finishing the proof.
\EP

We now proceed with bounding $|\sum_i\De_i|$, which follows Appendix~A.1 of \cite{AlS16} and resembles the standard multiplicative Chernoff bound proof. We bound $\Pr[\sum_i\De_i > a]$ as follows:
\begin{align*}
\Pr[\textstyle\sum_i\De_i > a]
&= \Pr[\exp(\la\textstyle\sum_i\De_i) > \exp(\la a)]
\\&< e^{-\la a} \E[\exp(\la \textstyle\sum_i\De_i)]
\\&= e^{-\la a} \Prod_i \E\lb e^{\la \De_i} \big\vert \De_1,\lds,\De_{i-1}\rb
\\& \stackrel{\mathclap{\text{Clm}\,\ref{clm:z}}}\le e^{-\la a} \prod_i \lp e^{-\la p} \lb pe^\la+(1-p)\rb \rp
\\&= e^{-\la a} e^{-\la pn} \lb pe^\la+(1-p)\rb^n.
\end{align*}
Similarly, we bound $\Pr[\sum_i\De_i < -a]$ as follows:
\begin{align*}
\Pr[\textstyle\sum_i\De_i <- a]
&= \Pr[\exp(-\la\textstyle\sum_i\De_i) > \exp(\la a)]
\\&< e^{-\la a} \E[\exp(-\la \textstyle\sum_i\De_i)]
\\&= e^{-\la a} \Prod_i \E\lb e^{-\la \De_i} \big\vert \De_1,\lds,\De_{i-1}\rb
\\& \stackrel{\mathclap{\text{Clm}\,\ref{clm:z}}}\le e^{-\la a} \prod_i \lp e^{\la p} \lb pe^{-\la}+(1-p)\rb \rp
\\&= e^{-\la a} e^{\la pn} \lb pe^{-\la}+(1-p)\rb^n,
\end{align*}
where the application of \Cref{clm:z} uses $-\la$ in place of $\la$ this time around.

It remains to use approximations to simplify the expressions. Instead of repeating them here, we refer the reader to the same approximations that precede Corollary~A.1.14 of Appendix~A.1 of \cite{AlS16}. We arrive at an analogous statement to Corollary~A.1.14 for martingales:
\BCL
$\Pr[|g(X) - \E[g(X)] > \e pn] < 2e^{-c_\e p n}$, for $c_\e := \min \{ -\ln(e^\e (1+\e)^{-(1+\e)}),\e^2/2 \}$.
\ECL
Finally, it is easy to check that for all $\e \in [0,1]$, $ -\ln(e^\e (1+\e)^{-(1+\e)}) \ge \e^2/3$. Thus, $c_\e \ge \e^2/3$, concluding \Cref{thm:m1}.

\Martingale*

\BP
Define $X':=X/b=(x_1/b,x_2/b,\lds,x_n/b)$, so that every coordinate of $X'$ is independently $1$ with probability $p$ and $0$ with probability $1-p$. Define $g'(Y):=g(bY)$, which is $bL$-Lipschitz, which means $g'/(bL)$ is $1$-Lipschitz. Observe that $g'(X')=g'(X/b)=g(X)$, and
\[ \left| \f{g'(X')}{bL}-\E\lb \f{g'(X')}{bL}\rb\right|\ge\e pn \iff \big| g(X)-\E\big[g(X)\big]\big|\ge \e pnbL .\]

Therefore, it suffices to prove that
\[ \Pr\lb \left| \f{g'(X')}{bL}-\E\lb \f{g'(X')}{bL}\rb\right|\ge\e pn \rb \le e^{-\e^2pn/3} .\] 
This follows directly from \Cref{thm:m1} on $X'$ and $g'/(bL)$, as desired.
\EP

\subsection{Proof of \Cref{clm:CoinFlip}}\label{sec:CoinFlip}

\CoinFlip*
\BP
Let $k:= pN$ and $\el:=N-k$.  The probability can be approximated by Stirling's formula, as follows:
\begin{align*}
\bn Nk p^k(1-p)^{N-k} &= \f{N!}{k!\el!} p^k(1-p)^{\el}
\\&\sim \f {(2\pi N)^{-1/2}(N/e)^N} {(2\pi k)^{-1/2}(k/e)^k \cd (2\pi\el)^{-1/2}(\el/e)^{\el}} p^k(1-p)^{\el}
\\&= \lp \f N{2\pi k\el} \rp^{1/2} \cd N^N \lp \f pk\rp^k \lp\f{1-p}\el\rp^\el
\\& = \lp \f N{2\pi (pN)\el} \rp^{1/2} \cd N^N \lp \f p{pN}\rp^k \lp \f{1-p}{N-pN}\rp^\el
\\&=\lp \f N{2\pi (pN)\el} \rp^{1/2} = \lp \f1{2\pi p \el} \rp^{1/2} \ge \lp\f1{2\pi N}\rp^{1/2},
\end{align*}
where the last inequality uses $p\le1$ and $\el\le N$.
\EP

\end{document}